\documentclass[12pt]{iopart}
\usepackage{iopams,setstack}
\usepackage{pifont,stmaryrd}
\usepackage{amssymb,amsthm}
\usepackage{enumerate,enumitem,easylist}
\usepackage{graphicx,float}
\usepackage[colorlinks=true,linkcolor=blue,citecolor=blue,urlcolor=blue]{hyperref}
\usepackage{color}
\usepackage{mathrsfs}

\newcommand{\dsp}{\displaystyle}

\newcommand{\Rbb}{\mathbb{R}}
\newcommand{\Cbb}{\mathbb{C}}

\newcommand{\repa}{\ensuremath{\Re\mathrm{e}}}

\newcommand{\adj}{\ensuremath{\mathrm{adj}}}
\newcommand{\res}{\ensuremath{\mathrm{res}}}
\newcommand{\spec}{\ensuremath{\mathrm{sp}}}
\newcommand{\spanvec}{\ensuremath{\mathrm{span}}}

\newcommand{\lap}{\mathcal{L}}
\newcommand{\mero}{\mathcal{M}}

\newcommand{\Cfrak}{\mathscr{C}}
\newcommand{\Mfrak}{\mathscr{M}}
\newcommand{\Pfrak}{\mathscr{P}}
\newcommand{\Qfrak}{\mathscr{Q}}
\newcommand{\Wfrak}{\mathscr{W}}

\newcommand{\xa}{x_a}
\newcommand{\xb}{x_b}
\newcommand{\xc}{x_c}
\newcommand{\xd}{x_d}
\newcommand{\xe}{x_e}

\newcommand{\xap}{x^+_a}
\newcommand{\xbp}{x^+_b}
\newcommand{\xcp}{x^+_c}
\newcommand{\xdp}{x^+_d}
\newcommand{\xep}{x^+_e}

\newcommand{\xam}{x^-_a}
\newcommand{\xbm}{x^-_b}
\newcommand{\xcm}{x^-_c}
\newcommand{\xdm}{x^-_d}
\newcommand{\xem}{x^-_e}

\newcommand{\xapm}{x^\pm_a}
\newcommand{\xbpm}{x^\pm_b}
\newcommand{\xcpm}{x^\pm_c}
\newcommand{\xdpm}{x^\pm_d}
\newcommand{\xepm}{x^\pm_e}

\newcommand{\xamp}{x^\mp_a}

\newcommand{\eone}{e_1}
\newcommand{\ep}{e_p}

\newcommand{\diffC}{\dot C}
\newcommand{\tC}{\tilde{C}}

\newcommand{\intintonen}{\{1,\dots,n \}}
\newcommand{\intintonentwo}{\{1,\dots,n^2 \}}

\newcommand{\ppmod}[1]{{\!\!\!\!\pmod{#1}}}

\newtheorem{theorem}{Theorem}[section]
\newtheorem{lemma}[theorem]{Lemma}
\newtheorem{corollary}[theorem]{Corollary}
\newtheorem{definition}[theorem]{Definition}
\newtheorem{remark}[theorem]{Remark}

\makeatletter
\newcounter{parentequation}% Counter for ``parent equation''.
\@ifundefined{ignorespacesafterend}{%
  \def\ignorespacesafterend{\global\@ignoretrue}%
}{}
\newenvironment{subequations}{%
  \refstepcounter{equation}%
  \protected@edef\theparentequation{\theequation}%
  \setcounter{parentequation}{\value{equation}}%
  \setcounter{equation}{0}%
  \def\theequation{\theparentequation\alph{equation}}%
  \ignorespaces
}{%
  \setcounter{equation}{\value{parentequation}}%
  \ignorespacesafterend
}
\makeatother

\newcommand{\be}{\begin{equation}}
\newcommand{\ee}{\end{equation}}

\newcommand{\bes}{\begin{equation*}}
\newcommand{\ees}{\end{equation*}}

\newcommand{\bea}{\begin{eqnarray}}
\newcommand{\eea}{\end{eqnarray}}

\newcommand{\beas}{\begin{eqnarray*}}
\newcommand{\eeas}{\end{eqnarray*}}

\newcommand{\bse}{\begin{subequations}}
\newcommand{\ese}{\end{subequations}}

\begin{document}
\title[Compartmental analysis identifiability]{Compartmental analysis of dynamic
nuclear medicine data: models and identifiability}
\author{Fabrice Delbary$^1$, Sara Garbarino$^2$ and Valentina Vivaldi$^2$}

\address{$^1$ Institut f\"{u}r Mathematik, Johannes Gutenberg-Universit\"at Mainz \\
$^2$ Dipartimento di Matematica and CNR-SPIN, Genova}

\begin{abstract}
Compartmental models based on tracer mass balance are extensively used in clinical and pre-clinical nuclear medicine in order to obtain quantitative information on tracer metabolism in the biological tissue. This paper is the first of a series of two that deal with the problem of tracer coefficient estimation via compartmental modelling in an inverse problem framework. Specifically, here we discuss the identifiability problem for a general n-dimension compartmental system and provide uniqueness results in the case of two-compartment and three-compartment compartmental models. The second paper will utilize this framework in order to show how non-linear regularization schemes can be applied to obtain numerical estimates of the tracer coefficients in the case of nuclear medicine data corresponding to brain, liver and kidney physiology.
\end{abstract}
%\maketitle

\section{Introduction}

Nuclear medicine imaging is a class of functional imaging modality that utilizes radioactive tracers to investigate specific physiological processes. 
Such tracers are in general short-lived isotopes that are injected in the subject's blood and linked to chemical compounds whose metabolism is highly significant to understand the function or malfunction of an organ. Positron 
Emission Tomography (PET) \cite{oletal97} is the most modern nuclear medicine technique, utilizing isotopes produced in a cyclotron and providing dynamical images of its metabolism-based accumulation in the tissues. While decaying, the isotope emits positrons that annihilate with the electrons of the tissue thus emitting two collimated gamma rays. These rays are detected by the PET collimators to provide a rather precise indication of their temporal and spatial origin. Applications of PET in the clinical workflow depend on the kind of tracer employed and on the kind of metabolism that such tracer is able to involve. For example, in oncological applications, $[^{18}F]$FDG \cite{adetal98,dietal98,koetal12} and $FMISO$ \cite{waetal09} are the most commonly used tracers; neuroimaging studies of Alzheimer disease utilizes 
$^{11}C$ and $^{15}O$ \cite{zhetal07}, while 
for myocardium perfusion analysis, $^{82}Rb$, $[^{18}F]$FDG  and $H_2$$^{15}O$ are the most extensively used tracers \cite{baetal06, beetal84, rietal88, scetal02}.
\par
From a computational viewpoint, PET experiments involve two kinds of inverse problems. 
In the first one, image reconstruction techniques are applied to reconstruct the 
spatiotemporal location of tracer concentration from the radioactivity measured 
by the detectors \cite{natterer01,huetal94,shetal82}.
The second problem utilizes these reconstructed dynamic PET data to estimate 
physiological parameters that describe the functional behaviour 
of the inspected tissues and therefore the flow of tracer between their 
different constituents. 
It is also possible to solve the full inverse problem of retrieving the compartment modelling coefficients straight from dynamic PET data. This problem has become increasingly interesting in recent years, and it is typically referred to as a \emph{direct} reconstruction problem \cite{beetal08, beetal10, reetal06, reetal07, yaetal08}; the one of splitting into two different separate problems is, by symmetry, referred to as \emph{indirect}.
The present paper is the first one of a series of two 
that focus on the indirect reconstruction problem and aim to describe it  
in an ill-posed inverse problems framework.

\par
Models in pharmacokinetics 
\cite{rescigno97, rescigno09} typically assume that in the organ under investigation there co-exist functionally separated pools of tracer, named {\em{compartments}}, that can exchange tracer between each other. With the help of the global observation of the organ along time provided by reconstructed PET images, {\em{compartmental analysis}} \cite{evetal05, gaetal13, gaetal14, gaetal15, guetal01, guetal02} aims at retrieving information on the radioactive tracer exchange rates between compartments. From a mathematical viewpoint, the time dependent concentrations of tracer in each compartment constitutes the state variables that can be determined from PET data and the time evolution of the state variables can be modelled by a system of differential equations for the concentrations, expressing the principle of tracer balance during exchange processes. Assuming that the exchange rates are time independent, and neglecting the spatial exchanges between compartment, the mathematical model for the compartmental problem becomes a linear system of Ordinary Differential Equations (ODEs) with constant coefficients. Although it is certainly possible to take into account macroscopic flow conditions (as particularly useful for modelling cardiac perfusion, for instance) and introduce a PDE-based framework, as in \cite{reetal14}, in this paper we will focus on discussing results for compartmental modelling under the standard and simplifying conditions of time independence of parameters and no spatial exchanges between compartment \cite{weaa04}.

This paper describes the analytical properties of this forward problem in the case of the 
general $n$-compartment system and in the more specific (but highly realistic) 
cases of the two-compartment and three-compartment catenary systems 
(where a catenary system refers to system made of chain of compartments, each one connected only to its immediate predecessor and successor in the chain). For all 
these compartmental problems, the constant coefficients describe the 
input/output rate of tracer for each compartment and represent the physiological 
parameters assessing the system's metabolism. Therefore such coefficients are 
the unknowns to be estimated in compartmental analysis inverse problem. 
Some results \cite{meshkat14} on the identifiability of these coefficients have been obtained recently, for a very specific class of compartment models, by means of graph analysis techniques and a reparametrizations procedure. In this paper, we will focus on an inverse problem approach for the identifiability, and provide the first general discussion of uniqueness for this inverse problem by proving some identifiability results in the case of the two-compartment and three-compartment catenary models. 
In a future paper we will provide a general 
scheme for the numerical solution of these inverse problems and apply it against 
both synthetic data and experimental measurements acquired by means of a PET 
system for small animals.

\par
The plan of the present paper is as follows. Section 2 will introduce the 
formalism for the most general $n$-th compartment model and for its 
specialization to the catenary case. Section 3 and Section 4 will provide the 
identifiability results for the two-compartment and three-compartment catenary 
models. Our conclusions will be offered in Section 5.

\section{The $n$-compartment systems}
\begin{figure}[H]
\begin{center}
\includegraphics[scale=0.625]{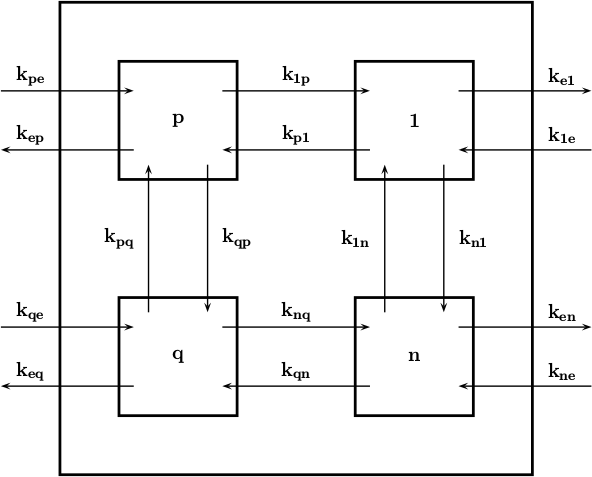}
\caption{Generic $n$-compartment system; compartments are labelled $p\in\intintonen$ and for each $p,q\in\intintonen$, $k_{pq}$ is the constant, non-negative rate of tracer exchange between compartment $p$ and $q$. The values $k_{ep}$ and $k_{pe}$ denotes the rates at which the tracer is received in the system at compartment $p$ and excreted by the system at compartment $p$, respectively.}\label{n-compartment-model}
\end{center}
\end{figure}
\hspace{6mm}
We start our analysis by considering the very general $n$-compartment model depicted in Figure 
\ref{n-compartment-model}. In all nuclear medicine modalities, the tracer 
injected into the body may assume different metabolic status. Since the spatial 
resolution provided by these imaging modalities is not sufficient to 
anatomically distinguish these different status, they are modelled as 
compartments, each one just playing a functional role that describes a specific 
tracer metabolic condition. We denote with $C_p$ the non-negative concentration 
function of the tracer in the compartment $p\in\intintonen$. We also assume that 
the compartment $p$ receives the radioactive tracer from outside the 
compartmental system with a concentration function $C_{p\mathrm{e}}$ and at a 
constant non-negative rate $k_{p\mathrm{e}}$ and excretes the tracer from inside 
at a constant non-negative rate $k_{\mathrm{e}p}$.  Further, the constant 
non-negative rate at which the compartment $p$ receives the tracer from a 
compartment $q\neq p$ is denoted with $k_{pq}$ and, finally, all input 
concentration functions $(C_{p\mathrm{e}})_{p\in\intintonen}$ are supposed to be 
non-negative and continuous. Then the evolution of the tracer concentrations in 
each compartment is governed by the following linear system of Ordinary 
Differential Equations (ODE) with constant coefficients
\bes
\diffC_p=\sum_{q=1}^n k_{pq}C_q+k_{p\mathrm{e}}C_{p\mathrm{e}},\qquad 
p\in\intintonen,
\ees
with the initial conditions
\be
C_p(0)=0,\qquad p\in\intintonen,
\ee
where, for $p\in\intintonen$, $k_{pp}=-\left(\sum_{q\neq 
p}k_{qp}+k_{\mathrm{e}p}\right)$. Dependence on $t$ is omitted but 
implied. That is
\be\label{veceqC}
\diffC=MC+W,\qquad C(0)=0, 
\ee
where
\be\label{vecCvecrhs}
C=\pmatrix{C_1\cr
\vdots\cr
C_n\cr},
\qquad
W=\pmatrix{k_{1\mathrm{e}}C_{1\mathrm{e}}\cr
\vdots\cr
k_{n\mathrm{e}}C_{n\mathrm{e}}\cr},
\ee
and the matrix $M$ is given by
\be\label{sysmat}
M_{pq}=k_{pq},\qquad p,q\in\intintonen.
\ee
PET-scan images allow to measure the total amount of radioactive tracer in the tissue or organ of interest, which is, in turn, modelled by the compartmental system. We assume that each compartment contributes to the intensity of the PET image linearly with respect to the amount of tracer in the compartment. Hence PET-scan data give access to the data $\tC$ such that the measurements and the model are related by the functional equation
\be\label{eqCtilde}
\tC(t)=\alpha^T C(t),\qquad t\in\Rbb_+,
\ee
where $\alpha\in{\Rbb_+^*}^n$ is assumed to be a known constant vector, representing the blood fraction in the tissue or organ of interest. Of course, the capability to reliably measure the blood fraction is a key point in compartmental analysis. We are aware that there are methods to measure it (as, for instance, in \cite{itetal01}) or to estimate it, as another model parameter in the compartmental model fitting (for instance using Levenberg-Marquardt as in \cite{waetal09}, or Maximum-Likelihood as in \cite{gaetal14}). For sake of simplicity in this paper, we assume it to be known \emph{a priori}, measured in advance by some experimental procedure, or fixed as the standard value retrieved in literature.

The general $n$-compartmental inverse problem is the 
one of recovering the exchange rates 
$\underline{K}\in\Rbb^{n^2+n}$, where
\bes
\underline{K}=\cases{k_{\mathrm{e}p},&$p\in\intintonentwo,\;p\equiv1\ppmod{n+1
}$,\\
k_{p-n\left\lfloor\frac{p-1}{n}\right\rfloor,1+\left\lfloor\frac{p-1}{n}
\right\rfloor},&$p\in\intintonentwo,\;p\not\equiv1\ppmod{n+1}$,\\
k_{(p-n^2)\mathrm{e}}&$p\in\{n^2+1,\ldots,n^2+n\}$,}
\ees
using measurements of $\tC$. 

\subsection{Properties of $n$-compartment systems}\label{genprop}
In this subsection, we propose to recall some noteworthy properties of $n$-compartment systems. In the following of the document, for a positive integer $n$ and $K\in\Rbb^{n^2+n}$, we denote by $\hat{K}\in\Rbb^{n^2}$ the first $n^2$ components of $K$ and $\check{K}\in\Rbb^n$ the last $n$ components of $K$. For a positive integer $n$, we denote by $\Mfrak$ the following linear 
operator
\be\label{Mfrak}
\begin{array}{rcl}
\Mfrak:\Rbb^{n^2}&\to&M_n(\Rbb)\\\
H&\mapsto&\Mfrak(H),
\end{array}
\ee
where for all $H\in\Rbb^{n^2}$
\bes
\fl\Mfrak(H)_{pq}=\cases{-H_{1+(n+1)(p-1)}-\sum_{\scriptsize\begin{array}{l}
p'=1\\p'\neq p\end{array}}^nH_{p+n(p'-1)},&$p,q\in\intintonen,\; p=q$,\\
H_{p+n(q-1)},&$p,q\in\intintonen,\;p\neq q$,}
\ees
so that for all $H\in\Rbb_+^{n^2}$, $\Mfrak(H)$ is the matrix defined in 
\eref{sysmat} for the parameters $H$.
First of all, for general $n$-compartment systems, we have the following 
theorem \cite{hearon63}
\begin{theorem}\label{hearonth1}
Consider $H\in\Rbb_+^{n^2}$. Then the eigenvalues of the matrix $M=\Mfrak(H)$ 
as defined in \eref{Mfrak} have a non-positive real part and if an eigenvalue 
has a zero real part, then the eigenvalue is $0$, moreover, $\dim(\ker(M))=m_0$ 
where $m_0$ is the multiplicity of $0$. In addition, the solution $C$ to
\be\label{solgeneq}
\dot{C}=MC+W,\qquad C(0)=C_0, 
\ee
where $C_0\in\Rbb_+^n$ and $W:\Rbb_+\to\Rbb_+^n$, verifies 
$C_p(t)\geq0$ for all $p\in\intintonen$ and $t\in\Rbb_+$.
\end{theorem}
\begin{remark}
As remarked in \cite{hearon63}, it is the principle of the conservation of mass applied to the system \eref{solgeneq} where $W=0$, which insures that: the eigenvalues of $M$ have a non-positive real part; the multiplicity of the eigenvalue $0$ is the dimension of the null--space of $M$ (the solutions are bounded); and the only possible eigenvalue with a zero real part is $0$ (if oscillations occur, then they are damped). The positiveness of $C$ is simply the fact that the concentrations are positive quantities.
\end{remark}
Some additional properties on the system matrix $M$ lead to more restrictions 
on its eigenvalues. In particular, we have \cite{hearon63}
\begin{theorem}\label{hearonth2}
Consider $H\in\Rbb_+^{n^2}$ and denote by $M=\Mfrak(H)$. If $M$ is irreducible, 
then $0$ is an eigenvalue of $M$ if and only if $k_{\mathrm{e}p}=0$ for 
$p\in\intintonen$, where for $p\in\intintonen$, 
$k_{\mathrm{e}p}=H_{1+(n+1)(p-1)}$ (no excretion). In other terms, $0$ is an 
eigenvalue of $M$ if and only if the $n$-compartment system of exchange rates 
$H\in\Rbb_+^{n^2}$ and without input is closed. Moreover, in that case, $0$ is 
a simple eigenvalue.
\end{theorem}
\begin{remark}
As remarked in \cite{hearon63}, a compartment system without input refers to a system in which all the external inputs are excluded (set to $0$). 
\end{remark}
\begin{theorem}\label{hearonth3}
Consider $H\in\Rbb_+^{n^2}$, denote by $M=\Mfrak(H)$ and suppose that there 
exist $(a_p)_{p\in\intintonen}$ in ${\Rbb_+^*}^n$ such that for all 
$p,q\in\intintonen,p\neq q$
\be\label{detbal}
k_{pq}a_q=k_{qp}a_p,
\ee
where for all $p,q\in\intintonen,p\neq q$, $k_{pq}=H_{p+n(q-1)}$. Then $M$ is 
diagonalizable and its eigenvalues are real and non-positive. 
\end{theorem}
\begin{remark}
Note that the condition \eref{detbal} does not depend on the excretion rates 
$k_{\mathrm{e}p}=H_{1+(n+1)(p-1)}$ for $p\in\intintonen$. Further, as remarked in \cite{hearon63}, the equation \eref{detbal} is the principle of detailed balance \cite{hearon53, hearon63, onsager31} for the closed $n$-compartment system of exchange rates $(k_{pq})_{p,q\in\intintonen,p\neq q}$: at the equilibrium state, every process is balanced by its inverse.
\end{remark}
As defined in \cite{hearon63, parter60, schmidt99}, we recall that a connected compartmental system is a system for which it is possible for the tracer to reach every compartment from every other compartment, that a compartmental system with no cycle is a system for which it is not possible for the tracer to pass from a given compartment through two or more other compartments back to the starting compartment and that a compartmental system is sign-symmetric if the associated matrix M is sign-symmetric: $k_{pq}k_{qp}\geq 0$ for $p\neq q$ and $k_{pq}=0$ if and only if $k_{qp}=0$. The theorem \ref{hearonth3} applies in particular to sign-symmetric systems with no cycle. In case of connected sign-symmetric systems with no cycle, like the catenary compartmental system we will use in the following, we also have the following result \cite{parter60}.
\begin{theorem}\label{parterth}
Consider a sign-symmetric cycle-free connected $n$-compartment system. Then an 
eigenvalue $\lambda$ of the system matrix $M$ is multiple if and only if there 
exists a compartment $p$, directly connected to at least three other 
compartments $p_1,p_2,p_3$, such that $\lambda$ is an eigenvalue of the 
matrices 
$N_1,N_2,N_3$ of the respective connected subsystems containing $p_1,p_2,p_3$ 
and not $p$.
\end{theorem}
\subsection{The $n$-compartment catenary case}
\begin{figure}[H]
\begin{center}
\includegraphics[scale=0.625]{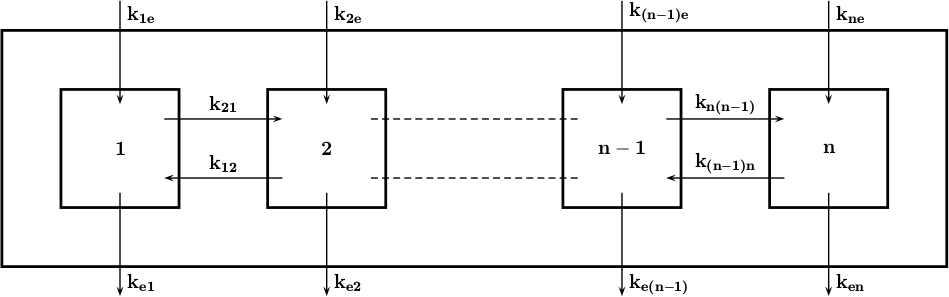}
\caption{Generic $n$-compartment catenary system; compartments are labelled $p\in\intintonen$, and the exchange rates accordingly, following the usual terminology.}
\end{center}
\end{figure}
A $n$-compartment catenary system is a $n$-compartment system such that
\beas
k_{\mathrm{e}p}\geq0,\qquad&p\in\intintonen,\\
k_{pq}>0,&p,q\in\intintonen,|p-q|=1,\\
k_{pq}=0,&p,q\in\intintonen,|p-q|>1,\\
k_{p\mathrm{e}}\geq0,&p\in\intintonen.
\eeas
According to the theorems \ref{hearonth1}, \ref{hearonth2}, \ref{hearonth3}, 
\ref{parterth} of the previous subsection, we have the following 
theorem
\begin{theorem}\label{simpnegeig}
The matrix $M$ of a $n$-compartment catenary system is diagonalizable and its 
eigenvalues are real, non-positive and simple. Moreover, $0$ is an eigenvalue 
of $M$ if and only if the system with no input is closed, that is 
$k_{\mathrm{e}p}=0$ for all $p\in\intintonen$.
\end{theorem}

\subsection{Identifiability definition}

We now recall the definition of identifiability \cite{evetal05,mietal11}. In order to 
introduce it, we need to denote by $\Wfrak$ the following linear operator
\bes
\begin{array}{rcl}
\Wfrak:C^0(\Rbb_+,\Rbb)^n&\to&L(\Rbb^n,C^0(\Rbb_+,\Rbb)^n)\\
\check{C}&\mapsto&\Wfrak(\check{C}),
\end{array}
\ees
where for all $\check{C}\in C^0(\Rbb_+,\Rbb)^n$, $H\in\Rbb^n$ and 
$p\in\intintonen$, we have $[\Wfrak(\check{C})(H)]_p=H_p\check{C}_p$, so that for all 
vector of input concentrations functions 
$\check{C}=(C_{p\mathrm{e}})_{p\in\intintonen}\in C^0(\Rbb_+,\Rbb_+)^n$ and 
$H\in\Rbb_+^n$, $\Wfrak(\check{C})(H)$ is the vector defined in 
\eref{vecCvecrhs} for the input concentrations functions $\check{C}$ and the 
parameters $H$. We denote by $\Cfrak$ the following function
\bes
\begin{array}{rcl}
\Cfrak:C^0(\Rbb_+,\Rbb)^n&\to&L(\Rbb^{n^2+n},C^1(\Rbb_+,\Rbb)^n)\\
\check{C}&\mapsto&\Cfrak(\check{C}),
\end{array}
\ees
where for all $\check{C}\in C^0(\Rbb_+,\Rbb)^n$ and $K\in\Rbb^{n^2+n}$, 
$C=\Cfrak(\check{C})(K)\in C^1(\Rbb_+,\Rbb)^n$ is the unique solution to
\bes
\diffC=MC+W,\qquad C(0)=0,
\ees
where $M=\Mfrak(\hat{K})$ and $W=\Wfrak(\check{C})(\check{K})$. For 
$\alpha\in{\Rbb^*_+}^n$, $\Cfrak^\alpha$ is the function defined by
\bes
\begin{array}{rcl}
\Cfrak^\alpha:C^0(\Rbb_+,\Rbb)^n&\to&L(\Rbb^{n^2+n},C^1(\Rbb_+,\Rbb))\\
\check{C}&\mapsto&\left[K\mapsto\alpha^T(\Cfrak(\check{C})(K))\right].
\end{array}
\ees
Consider a positive integer $n$, input concentration functions $\check{C}\in 
C^0(\Rbb_+,\Rbb_+)^n$, $\alpha\in{\Rbb^*_+}^n$ and a subset $\Omega$ of 
$\Rbb_+^{n^2+n}$ of admissible exchange rates $K$. Define 
$\tC=\Cfrak^\alpha(\check{C})$ and denote by $\tC_\Omega$ the restriction of $\tC$ to 
$\Omega$.
\begin{definition}\hspace{\linewidth}
\begin{enumerate}
\item The model of equations \eref{veceqC}-\eref{eqCtilde} is said globally 
identifiable at $K\in\Omega$, if $\tC_\Omega^{-1}(\{\tC(K)\})=\{K\}$.
\item The model of equation \eref{veceqC}-\eref{eqCtilde} is said locally 
identifiable at $K\in\Omega$ if there exists $\epsilon>0$ such that 
$\tC_{\Omega,K,\epsilon}^{-1}(\{\tC(K)\})=\{K\}$ where 
$\tC_{\Omega,K,\epsilon}$ 
denotes the restriction of $\tC_\Omega$ to the open ball $B_{K,\epsilon}$ of 
$\Omega$, with centre $K$ and radius $\epsilon$.
\item The model of equation \eref{veceqC}-\eref{eqCtilde} is said structurally 
globally identifiable if it is globally identifiable at all $K\in\Omega$.
\item The model of equation \eref{veceqC}-\eref{eqCtilde} is said structurally 
locally identifiable if it is locally identifiable at all $K\in\Omega$.
\end{enumerate}
\end{definition}

When we have a general $n$-compartment compartmental system, it is hard to find 
a precise characterization of identifiability; it is although possible to prove some weak results, as the following lemma,
leading, in particular cases, to more precise identifiability results.

We consider (from now on) $n$-compartment systems where for 
$p\in\intintonen$, the input concentration functions $C_{p\mathrm{e}}$ from 
$\Rbb_+$ to $\Rbb_+$ are Laplace-transformable. For $p\in\intintonen$, 
$r_{p\mathrm{e}}=\inf\left\{r\in\Rbb:\int_0^{+\infty}e^{-rt}C_{p\mathrm{e}}(t)\,
dt<+\infty\right\}$ will denote the abscissa of convergence of 
$C_{p\mathrm{e}}$. Note that in practical applications, the input concentration 
functions are bounded so that they are Laplace-transformable and for 
$p\in\intintonen$, $r_{p\mathrm{e}}\leq0$. The Laplace transform will be 
denoted 
by $\lap$ and for $r\in\Rbb$, we define $\Cbb_r=\{z\in\Cbb:\repa(z)>r\}$.
\begin{lemma}\label{fracdec}
Consider a positive integer $n$, Laplace-transformable input concentration 
functions $\check{C}=(C_{p\mathrm{e}})_{p\in\intintonen}\in 
C^0(\Rbb_+,\Rbb_+)^n$, $\alpha\in{\Rbb^*_+}^n$ and exchange rates 
$K\in\Rbb_+^{n^2+n}$. Then, for $\alpha\in{\Rbb_+^*}^n$, 
$\tC=\Cfrak^\alpha(\check{C})(K)$ is Laplace-transformable and its abscissa of 
convergence $r$ verifies
\be\label{absco}
r\leq 
r_{\mathrm{m}}=\max\left(\{r^{\mathrm{e}}_{\mathrm{m}}\}\cup\spec(M)\right),
\ee
where
\be\label{abscoe}
r^{\mathrm{e}}_{\mathrm{m}}=\max\left(\{r_{p\mathrm{e}}:p\in\intintonen\}\right)
,
\ee
and $\spec(M)$ denotes the spectrum of the matrix $M=\Mfrak(\hat{K})$. Moreover 
for all $z\in\Cbb_{r_{\mathrm{m}}}$, we have
\bes
\lap\tC(z)=\sum_{p=1}^n \alpha_pk_{p\mathrm{e}}\frac{Q_p(z)}{P(z)}\lap 
C_{p\mathrm{e}}(z),
\ees
where the unitary polynomial $P$ of degree $n$ is the characteristic polynomial 
of $M$ and for $p\in\intintonen$, $Q_p$ is a unitary polynomial of degree $n-1$ 
given by
\bes
Q_p(z)=\frac{\alpha^T\adj(zI_n-M)\ep}{\alpha_p},
\ees
where $\adj(zI_n-M)$ denotes the adjugate matrix of $zI_n-M$.
\end{lemma}
\begin{proof}
The first statement of the lemma is obvious. In addition, since 
$C=\Cfrak(\check{C})(K)$ verifies \eref{veceqC}, for all 
$z\in\Cbb_{r_{\mathrm{m}}}$, we have
\bes
\lap\tC(z)=\sum_{p=1}^nk_{p\mathrm{e}}\alpha^T(zI_n-M)^{-1}\ep\,\lap 
C_{p\mathrm{e}}(z).
\ees
Since
\bes
(zI_n-M)^{-1}=\frac{\adj(zI_n-M)}{\det(zI_n-M)}=\frac{\adj(zI_n-M)}{P(z)},
\ees
where $P$ is the characteristic polynomial of $M$, we have
\bes
\lap\tC(z)=\sum_{p=1}^nk_{p\mathrm{e}}\frac{\alpha^T\adj(zI_n-M)\ep}{P(z)}\lap 
C_{p\mathrm{e}}(z).
\ees
Moreover, it can be easily remarked that the only cofactors of $zI_n-M$ of 
degree 
$n-1$ in $z$ are the diagonal cofactors and that the coefficient of the 
monomial 
$z^{n-1}$ in these cofactors is $1$, so that for all $p\in\intintonen$, the 
polynomial $\alpha^T\adj(zI_n-M)\ep$ in $z$ is of degree $n-1$ and its leading 
coefficient is $\alpha_p$.
\end{proof}
For a positive integer $n$, we denote by $\Pfrak$ the function
\bes
\begin{array}{rcl}
\Pfrak:\Rbb^{n^2}&\to&\Rbb[X]\\
H&\mapsto&\det(XI_n-\Mfrak(H)),
\end{array}
\ees
and for $\alpha\in{\Rbb_+^*}^n$, we denote by $\Qfrak^\alpha$ the function
\bes
\begin{array}{rcl}
\Qfrak^\alpha:\Rbb^{n^2}&\to&\Rbb[X]^n\\
H&\mapsto&\dsp\left(\frac{\alpha^T\adj(XI_n-\Mfrak(H))\ep}{\alpha_p}\right)_{
p\in\intintonen}.
\end{array}
\ees
For $r\in\Rbb$, we denote by $\mero(\Cbb_r)$ the field of meromorphic functions 
on $\Cbb_r$ and by $\Rbb(X)$ the field of rational fractions in the 
indeterminate $X$ with coefficients in $\Rbb$. In the following, for 
$r\in\Rbb$, 
$\mero(\Cbb_r)$ will be regarded as a vector space over $\Rbb(X)$. As an 
immediate consequence of lemma \ref{fracdec}, we have
\begin{corollary}\label{coreqratfrac}
Consider a positive integer $n$ and Laplace-transformable input concentration 
functions $\check{C}=(C_{p\mathrm{e}})_{p\in\intintonen}\in 
C^0(\Rbb_+,\Rbb_+)^n$ such that $C_{p\mathrm{e}}=0$ for 
$p\in\intintonen\setminus\{p_\ell:\ell\in\{1,\dots,k\}\}$, where $k\geq1$ and 
$(p_\ell)_{\ell\in\{1,\dots,k\}}$ are $k$ distinct elements of $\intintonen$. 
We recall that for exchange rates $K\in\Rbb_+^{n2+n}$, $\check{K}$ denotes the 
last $n$ components of $K$, that is, for $p\in\intintonen$, 
$\check{K}_p=k_{p\mathrm{e}}$. Consider $\alpha\in{\Rbb_+^*}^n$ and define 
$\tC=\Cfrak^\alpha(\check{C})$. If the functions $(\lap 
C_{p_\ell\mathrm{e}})_{\ell\in\{1,\dots,k\}}$ of 
$\mero(\Cbb_{r^{\mathrm{e}}_{\mathrm{m}}})$ are linearly independent over 
$\Rbb(X)$, where $r^{\mathrm{e}}_{\mathrm{m}}$ is defined by \eref{abscoe}, 
then 
for $K\in\Rbb_+^{n^2+n}$, the exchange rates 
$(\check{K}_{p_\ell\mathrm{e}})_{\ell\in\{1,\dots,k\}}$ are uniquely determined 
by the function $\tC(K)$, that is, for all $K'\in\Rbb_+^{n^2+n}$, we have 
$\tC(K')=\tC(K)$ only if 
$\check{K}_{p_\ell\mathrm{e}}=\check{K}'_{p_\ell\mathrm{e}}$ for all 
$\ell\in\{1,\dots,k\}$. More precisely, we have $\tC(K')=\tC(K)$ if and only if 
the previous condition holds and for all $\ell\in\{1,\dots,k\}$
\bes
\frac{\Qfrak^\alpha(K')_{p_\ell}}{\Pfrak(K')}=\frac{\Qfrak^\alpha(K)_{p_\ell}}{
\Pfrak(K)}.
\ees
\end{corollary}
Otherwise, in a more general case, we have
\begin{corollary}
Consider a positive integer $n$ and Laplace-transformable input concentration 
functions $\check{C}=(C_{p\mathrm{e}})_{p\in\intintonen}\in 
C^0(\Rbb_+,\Rbb_+)^n$, then, there exist $k$ distinct elements 
$(p_\ell)_{\ell\in\{1,\dots,k\}}$ of $\intintonen$, where $k\in\{0,\dots,n\}$ 
($k=0$ if all input concentration functions are identically zero) such that the 
functions $(\lap C_{p_\ell\mathrm{e}})_{\ell\in\{1,\dots,k\}}$ of 
$\mero(\Cbb_{r^{\mathrm{e}}_{\mathrm{m}}})$ are linearly independent over 
$\Rbb(X)$, where $r^{\mathrm{e}}_{\mathrm{m}}$ is defined by \eref{abscoe}, and 
such that for all $p\in\intintonen\setminus\{p_\ell:\ell\in\{1,\dots,k\}\}$, 
$\lap C_{p\mathrm{e}}\in\spanvec(\{\lap 
C_{p_\ell\mathrm{e}}:\ell\in\{1,\dots,k\}\})$. Consider $\alpha\in{\Rbb_+^*}^n$ 
and denote by $\tC=\Cfrak^\alpha(\check{C})$. For all $K\in\Rbb^{n^2+n}$ and 
$z\in\Cbb_{r_\mathrm{m}}$, where $r_{\mathrm{m}}$ is defined by \eref{absco}, 
we 
then have
\bes
\lap\tC(z)=\sum_{\ell=1}^kF^{\check{C},K,\alpha}_\ell(z)\lap 
C_{p_\ell\mathrm{e}}(z),
\ees
where for all $\ell\in\{1,\dots,k\}$, $F^{\check{C},K,\alpha}_\ell\in\Rbb(X)$. 
Moreover, for $K,K'\in\Rbb^{n^2+n}$, we have $\tC(K)=\tC(K')$ if and only if 
for 
all $\ell\in\{1,\dots,k\}$
\bes
F^{\check{C},K',\alpha}_\ell=F^{\check{C},K,\alpha}_\ell.
\ees
\end{corollary}

These are general, weak, results on identifiability of $n$-compartment systems; in the next sections, we will focus our study of identifiability on the $2$-compartments and the $3$-compartments catenary systems, and find a precise characterization of identifiability.

\section{Study of a $2$-compartment catenary system}

\begin{figure}[H]
\begin{center}
\includegraphics[scale=0.75]{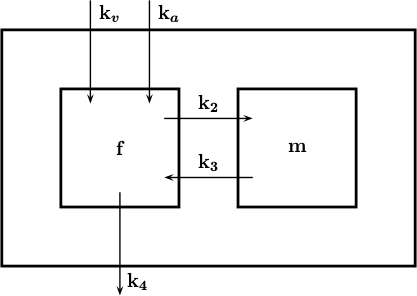}
\caption{Generic $2$-compartment catenary model.}\label{catenary-2-figure}
\end{center}
\end{figure}
We introduce now the case of a $2$-compartment catenary system. The low degree 
of complexity of this $2$-compartmental model, allows its daily clinical utilization, without high computational request. 
For this reason, this model gained high popularity over the last twenty years, and has been extensively used to analyze PET data. 
FDG-PET data are particularly favourable for the application of this model; it is indeed possible to simplify the metabolism of the tracer in 
two chemical phases: a free intracellular phase and a metabolized ones 
\cite{gaetal13, gaetal14, qietal07}, that well suit a $2$-compartment catenary model. We recall that the evolution equations in 
a $2$-compartment catenary system as 
in Figure \ref{catenary-2-figure} are given by
\beas
\diffC_1=-\left(k_{\mathrm{e}1}+k_{21}\right)C_1+k_\mathrm{12}C_2,\\
\diffC_2=k_{21}C_1-k_{12}C_2,
\eeas
with initial conditions
\bes
C_1(0)=0,\qquad C_2(0)=0.
\ees
That is
\be\label{vecorddiffeqC2}
\diffC=MC+W,\qquad C(0)=0,
\ee
where $M\in M_2(\Rbb),C\in C^1(\Rbb_+,\Rbb)^2,W\in C^0(\Rbb_+,\Rbb)^2$ are 
given by
\bes
\fl M=\pmatrix{-(k_{\mathrm{e}1}+k_{21})&k_{12}\cr
k_{21}&-k_{12}\cr},\qquad C=\pmatrix{C_1\cr
C_2\cr},\qquad W=k_{1\mathrm{e}}\pmatrix{C_{1\mathrm{e}}\cr
0\cr}=k_{1\mathrm{e}}C_{1\mathrm{e}}\eone.
\ees
PET-scan images allow to access $C^*=(1-V)(C_1+C_2)+VC_{1\mathrm{e}}$ where $V$ 
is a fraction in $(0,1)$, represent the blood fraction in the tissue 
under examination, and in this analysis, as already observed, it is assumed to be known. Since $C_{1\mathrm{e}}$ is the concentration of tracer in 
blood, it is directly measurable \cite{kuetal09, raetal13}. It is therefore possible 
to rewrite the previous equation as $\tC = (1-V)(C_1+C_2)$, where 
$\tC=C^*-VC_{1\mathrm{e}}$, which is measurable from PET-scans. The inverse 
problem consists in recovering $k_{1\mathrm{e}}$, $k_{\mathrm{e}1}$, $k_{21}$, 
$k_{12}$ from the knowledge of $V$, $C_{1\mathrm{e}}$ and
\be\label{eqCtilde2comp}
\tC=\alpha^T C,\qquad\mbox{where }\alpha=\pmatrix{1-V &
1-V\cr}.
\ee
More generally, we will study the identifiability of this $2$-compartment 
catenary system on $\Omega={\Rbb_+^*}^4$. For sake of simplicity, in the 
following of the document, the exchange rates of the system are denoted by 
$a,b,c,k$ where
\bes
a=k_{\mathrm{e}1},\qquad b=k_{21},\qquad c=k_{12},\qquad k=k_{1\mathrm{e}}.
\ees
$M$ and $W$ are then rewritten
\bes
M=\pmatrix{-(a+b)&c\cr
b&-c\cr},\qquad W=kC_{1\mathrm{e}}\eone.
\ees
We recall that according to theorem \ref{simpnegeig}, $M$ is diagonalizable and 
its eigenvalues are simple and negative. We suppose that $C_{1\mathrm{e}}$ is 
bounded and not identically zero. Hence, according to corollary 
\ref{coreqratfrac}, we first have the following lemma
\begin{lemma}
$k$ is uniquely determined by the knowledge of $\tC$.
\end{lemma}
Now, in order to know the solutions of the inverse problem, we have to search
for the matrices $M_x\in M_2(\Rbb)$
\bes
M_x=\pmatrix{-(\xa+\xb)&\xc\cr
\xb&-\xc\cr},
\ees
where $\xa,\xb,\xc\in\Rbb^*_+$, such that $F_x=F$, with $F=Q/P$ and 
$F_x=Q_x/P_x$, where $P$ is the characteristic polynomial of $M$, $P_x$ is the 
characteristic polynomial of $M_x$ and $Q,Q_x$ are defined by
\bes
Q(X)=\alpha^T\adj(X-M)\eone,\qquad Q_x(X)=\alpha^T\adj(X-M_x)\eone.
\ees
That is
\beas
P(X)=X^2+(a+b+c)X+ac,\\
Q(X)=(1-V)(X+b+c),\\
P_x(X)=X^2+(\xa+\xb+\xc)X+\xa\xc,\\
Q_x(X)=(1-V)(X+\xb+\xc).
\eeas
First of all, we write the rational fraction by splitting the polynomial $P$
\bes
F(X)=\frac{(1-V)(X+b+c)}{(X-\lambda_1)(X-\lambda_2)},
\ees
where $\lambda_1=-\frac{a+b+c+\sqrt{(a+b+c)^2-4ac}}{2}$ and 
$\lambda_2=-\frac{a+b+c-\sqrt{(a+b+c)^2-4ac}}{2}$. Remark that in this simple 
example, we can easily verify the statement of theorem \ref{simpnegeig}. Indeed, 
we 
have $0<(a+b-c)^2+4bc=(a+b+c)^2-4ac<a+b+c$ so that $\lambda_1$ and $\lambda_2$ 
are real, negative and distinct. Moreover, we have 
$\lambda_1,\lambda_2\neq-(b+c)$, hence, the rational fraction $F$ is in its 
irreducible form. Since $\deg P_x = \deg P =3$ and $\deg Q_x = \deg Q =2$, we have
$F_x=F$ only if the rational fraction $F_x$ is irreducible too (but as we have 
just seen, since $\xa,\xb,\xc>0$, it is always irreducible). In addition, since 
the leading coefficients of $P_x$ and $P$ are identical, as well as those of 
$Q_x$ and $Q$, we have $F_x=F$ if and only if $P_x=P$ and $Q_x=Q$, that is
\beas
\xa+\xb+\xc=a+b+c,\\
\xa\xc=ac,\\
\xb+\xc=b+c.
\eeas
Hence $F_x=F$ if and only if $\xa=a$, $\xb=b$, $\xc=c$, that is $M_x=M$. 
Consequently, we have the following theorem
\begin{theorem}
The model of equations \eref{vecorddiffeqC2}-\eref{eqCtilde2comp} is 
structurally globally identifiable.
\end{theorem}

\section{Study of a $3$-compartment catenary system}
\begin{figure}[H]
\begin{center}
\includegraphics[scale=0.75]{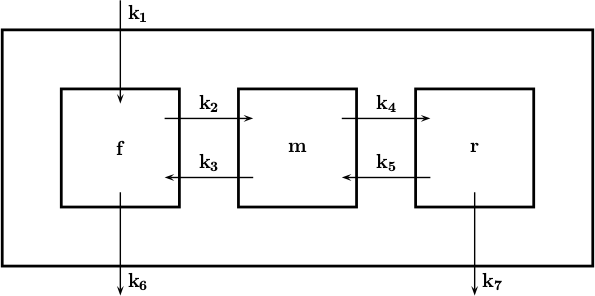}
\caption{Generic $3$-compartment catenary model.}\label{catenary-3-figure}
\end{center}
\end{figure}

We now introduce a $3$-compartment catenary system \cite{raetal13}. 
This model is commonly used either in FDG-PET studies (and the third compartment accounts for the presence of a blood stream that carries tracer inside the organ under examination, approximated by $C_{1e}$ in $2$-compartment models) or in $^{11}C$-PET studies, where the tracer has an intrinsic three-status metabolism \cite{koetal91}. We recall that the evolution equations in 
a $3$-compartment catenary system as in Figure \ref{catenary-3-figure} are 
given by
\beas
\diffC_1=-\left(k_{\mathrm{e}1}+k_{21}\right)C_1+k_\mathrm{12}C_2+k_{1\mathrm{e}
}C_{1\mathrm{e}},\\
\diffC_2=k_{21}C_1-\left(k_{12}+k_{32}\right)C_2+k_{23}C_3,\\
\diffC_3=k_{32}C_2-k_{23}C_3,
\eeas
with initial conditions
\bes
C_1(0)=0,\qquad C_2(0)=0,\qquad C_3(0)=0.
\ees
That is
\be\label{vecorddiffeqC3}
\diffC=MC+W,\qquad C(0)=0,
\ee
where $M\in M_3(\Rbb),C\in C^1(\Rbb_+,\Rbb)^3,W\in C^0(\Rbb_+,\Rbb)^3$ are given 
by
\beas
M=\pmatrix{-(k_{\mathrm{e}1}+k_{21})&k_{12}&0\cr
k_{21}&-(k_{12}+k_{32})&k_{23}\cr
0&k_{32}&-k_{23}\cr},\\
C=\pmatrix{C_1\cr
C_2\cr
C_3\cr},\qquad W=k_{1\mathrm{e}}\pmatrix{C_{1\mathrm{e}}\cr
0\cr
0\cr}=k_{1\mathrm{e}}C_{1\mathrm{e}}\eone.
\eeas
PET-scan images allow to access $\tC=VC_1+(1-V)\left(C_2+C_3\right)$ where $V$ 
is a fraction in $(0,1)$, representing the blood fraction in the tissue 
under examination, and in this analysis, as already observed, it is assumed to be known. The inverse problem consists in recovering $k_{1\mathrm{e}}$, 
$k_{\mathrm{e}1}$, $k_{21}$, $k_{12}$, $k_{32}$, $k_{23}$ from the knowledge of 
$V$, $C_{1\mathrm{e}}$ and
\be\label{eqCtilde3comp}
\tC=\alpha^T C,\qquad\mbox{where }\alpha=\pmatrix{V &
1-V &
1-V\cr}.
\ee
More generally, we will study the identifiability of this $3$-compartment 
catenary system on $\Omega={\Rbb_+^*}^6$. For sake of simplicity, in the 
following of the document, the exchange rates of the system are denoted by 
$a,b,c,d,e,k$ where
\bes
\fl a=k_{\mathrm{e}1},\qquad b=k_{21},\qquad c=k_{12},\qquad d=k_{32},\qquad 
e=k_{23},\qquad k=k_{1\mathrm{e}}.
\ees
$M$ and $W$ are then rewritten
\bes
M=\pmatrix{-(a+b)&c&0\cr
b&-(c+d)&e\cr
0&d&-e\cr},\qquad W=kC_{1\mathrm{e}}\eone.
\ees
We recall that according to theorem \ref{simpnegeig}, $M$ is diagonalizable and 
its eigenvalues are simple and negative. We suppose that $C_{1\mathrm{e}}$ is 
bounded and not identically zero. Hence, according to corollary 
\ref{coreqratfrac}, we first have the following lemma
\begin{lemma}
$k$ is uniquely determined by the knowledge of $\tC$.
\end{lemma}
Using the properties of the $n$-compartment catenary systems given in theorem \ref{simpnegeig} and the corollary \ref{coreqratfrac}, we can write the following theorem, whose proof can be found in the appendix \ref{appendix}.
\begin{theorem}\label{identifiability-3-theorem}
If $V\geq1/2$, the model of equations 
\eref{vecorddiffeqC3}-\eref{eqCtilde3comp} is structurally globally identifiable. Otherwise, the model is neither structurally globally nor structurally locally identifiable. However, the model is locally identifiable at points 
$(a,b,c,d,e,k)\in{\Rbb_+^*}^6$ such that
\bes
a\neq\frac{1-2V}{2V(1-V)}\left((c+d+e)V+(1-V)b\pm\sqrt{\Delta_Q}\right).
\ees
More precisely, in that case, the model is globally identifiable at points 
$(a,b,c,d,e,k)\in{\Rbb_+^*}^6$ such that one of the following exclusive 
conditions of lemma \ref{locunipqcoprime} holds: \ref{1}, \ref{2.1}, 
\ref{2.2.1}, \ref{2.2.2.1}, \ref{3.1.1}, \ref{3.1.2.1.1}, \ref{3.1.2.2}, 
\ref{3.2.1}, \ref{3.2.2.1}, \ref{3.2.2.2.1}, \ref{3.3.1.1}, \ref{3.3.2.1}. 
Otherwise, the inverse problem has two distinct solutions. In the case where
\bes
a=\frac{1-2V}{2V(1-V)}\left((c+d+e)V+(1-V)b\pm\sqrt{\Delta_Q}\right),
\ees
the model is neither globally nor locally identifiable. More precisely, the set 
of solutions of the inverse problem is a curve of $\Rbb^6$.
\end{theorem}

\section{Conclusions}
This paper describes the identifiability problem for a general $n$-compartment systems and provides some precise uniqueness results in the case of the two- and three-compartment catenary models. We think that the main advantage of our approach is its notable degree of generality, which allows a rather unified approach to the definition of the model and, particularly, to the proof of the uniqueness results. Further generalization of the identifiability results to n-compartment models, with $n>3$, are certainly difficult although some uniqueness conditions can be probably formulated in the case of low dimension non-catenary models and catenary models with higher dimension.
\par
A second paper concerned with two- and three-compartment catenary models is under preparation, which will address the inverse problem of numerically determining the tracer coefficients by means of a general Newton regularized scheme. In this paper applications concerning cerebral, hepatic, and renal functions will be considered, involving experimental measurements acquired by means of a PET system for small animals. Further investigation will deal with the formulation of a compartmental model in which the tracer coefficients are local parameters and therefore the numerical solution of the inverse problem leads to the construction of parametric images.

\section*{Acknowledgments}
This work was supported by a grant funded by Liguria Region, PO CRO FSE 2007/13 Asse IV for the project 'Optimization of spatial resolution for the analysis of physiological variables in small animals--Data acquisition and interpretation in micro-PET.' We thank the anonymous referees for their suggestions that helped to improve this work.
\section*{References}
\providecommand{\newblock}{}

\renewcommand{\thesection}{\Alph{section}}
\setcounter{section}{0}
\section{APPENDIX}\label{appendix}
In order to prove theorem \ref{identifiability-3-theorem} we first need to 
search for the matrices $M_x\in M_3(\Rbb)$
\bes
M_x=\pmatrix{-(\xa+\xb)&\xc&0\cr
\xb&-(\xc+\xd)&\xe\cr
0&\xd&-\xe\cr},
\ees
where $\xa,\xb,\xc,\xd,\xe\in\Rbb^*_+$, such that $F_x=F$, with $F=Q/P$ and 
$F_x=Q_x/P_x$, where $P$ is the characteristic polynomial of $M$, $P_x$ is the 
characteristic polynomial of $M_x$ and $Q,Q_x$ are defined by
\bes
Q(X)=\alpha^T\adj(X-M)\eone,\qquad Q_x(X)=\alpha^T\adj(X-M_x)\eone.
\ees
That is
\beas
\fl P(X)=X^3+(a+b+c+d+e)X^2+(ac+ad+ae+bd+be+ce)X+ace,\\
\fl Q(X)=VX^2+(V(c+d+e)+(1-V)b)X+(Vce+(1-V)(d+e)b),\\
\fl P_x(X)=X^3+(\xa+\xb+\xc+\xd+\xe)X^2\\
+(\xa\xc+\xa\xd+\xa\xe+\xb\xd+\xb\xe+\xc\xe)X+\xa\xc\xe,\\
\fl Q_x(X)=VX^2+(V(\xc+\xd+\xe)+(1-V)\xb)X+(V\xc\xe+(1-V)(\xd+\xe)\xb).
\eeas
We first begin with some preliminary results about the polynomials $P$ and $Q$ 
and the rational fraction $F$. The discriminant of $Q$ is given by
\beas
\Delta_Q&=(V(c+d+e)+(1-V)b)^2-4V(Vce+(1-V)(d+e)b),\\
&=((-c+d+e)V-(1-V)b)^2+4cdV^2,
\eeas
thus $0<\Delta_Q<V(c+d+e)+(1-V)b$. Hence $Q$ has two distinct real negative roots
\bes
\mu_q=-\frac{(c+d+e)V+(1-V)b+(-1)^q \sqrt{\Delta_Q}}{2V}, \qquad q\in\{1,2\},
\ees
Hence, the rational fraction $F$ can be written as
\bes
F(X)=\frac{V(X-\mu_1)(X-\mu_2)}{(X-\lambda_1)(X-\lambda_2)(X-\lambda_3)},
\ees
where $\lambda_1<\lambda_2<\lambda_3<0$ are the eigenvalues of $M$. In the 
following, we will need to know the irreducible form of the rational fraction 
$F$, that is, the number of roots common to $P$ and $Q$. For the moment, we have 
three possibilities.
\begin{enumerate}[label={\ding{\arabic*}},start=172]
\item$P$ and $Q$ have no common root, \textit{i.e.} they are coprime, so that 
$F$ is irreducible,
\item$P$ and $Q$ have one common root $\lambda_p=\mu_q, p\in\{1,2,3\}, 
q\in\{1,2\}$, hence
\bes 
\fl 
F=\frac{V(X-\mu_{q'})}{(X-\lambda_{p'})(X-\lambda_{p''})},\qquad\mbox{where } p',
p''\in\{1,2,3\}\setminus\{p\}, q'\in\{1,2\}\setminus\{q\},
\ees
\item$P$ and $Q$ have two common roots, \textit{i.e.} $Q|P$, 
$\mu_1=\lambda_p, \mu_2=\lambda_{p'}, p,p'\in\{1,2,3\}, p\neq p'$, hence
\bes 
F=\frac{V}{X-\lambda_{p''}},\qquad\mbox{where }p''\in\{1,2,3\}\setminus\{p,p'\}.
\ees
\end{enumerate}
We observe that $P$ and $Q$ are not coprime if and only if their resultant 
$\res(P,Q)$ is $0$. 
The resultant $\res(P,Q)$ of $P$ and $Q$ is given by
\beas
\fl\res(P,Q)=-b^2cd\left(V(1-V)^2a^2-(1-V)(1-2V)(V(c+d+e)+(1-V)b)a\right.\\
+\left.(1-2V)^2(Vce+(1-V)(d+e)b)\right).
\eeas
In particular, we can see that if $V=1/2$, $\res(P,Q)=-a^2b^2cd/8\neq 0$, 
otherwise, it can be easily remarked that
\bes
\res(P,Q)=-b^2cd(1-2V)^2Q\left(-\frac{1-V}{1-2V}a\right).
\ees
Hence, for $V\neq1/2$, $\res(P,Q)=0$ if and only if $-\frac{1-V}{1-2V}a$ is a 
root of $Q$. In particular, if $V>1/2$, then $-\frac{1-V}{1-2V}a>0$. However, 
the roots of $Q$ are negative, thus $\res(P,Q)\neq0$. We recall that the same 
result holds for $V=1/2$, thus we have the following lemma
\begin{lemma}\label{irpqvoh}
If $V\geq1/2$, then $P$ and $Q$ are coprime so that $F$ is irreducible.
\end{lemma}
If now $V<1/2$ and $-\frac{1-V}{1-2V}a$ is a root of $Q$, that is
\bes
a=\frac{1-2V}{2V(1-V)}\left((c+d+e)V+(1-V)b+(-1)^q \sqrt{\Delta_Q}\right), 
\qquad q\in\{1,2\},
\ees
then $P(X)=(X-\mu_q)\breve{P}(X)$, where
\bes
\fl\breve{P}(X)=X^2+\frac{(2-3V)(c+d+e)+(1-V)b-(-1)^q\sqrt{\Delta_Q}}{2(1-V)}
X+\frac{1-2V}{1-V}ce,
\ees
that is
\bes
\breve{P}(X)=X^2+\left(-\frac{V}{1-2V}a+b-\frac{1-2V}{1-V}(c+d+e)\right)X+\frac{
1-2V}{1-V}ce.
\ees
Since $Q(X)=(X-\mu_q)\breve{Q}(X)$, with
\bes
\breve{Q}(X)=V\left(X+\frac{(c+d+e)V+(1-V)b-(-1)^q \sqrt{\Delta_Q}}{2V}\right),
\ees
that is
\bes
\breve{Q}(X)=V\left(X-\frac{1-V}{1-2V}a+\frac{1-V}{V}b+c+d+e\right),
\ees
we have
\bes
F(X)=\frac{V\left(X-\frac{1-V}{1-2V}a+\frac{1-V}{V}b+c+d+e\right)}{
X^2+\left(-\frac{V}{1-2V}a+b-\frac{1-2V}{1-V}(c+d+e)\right)X+\frac{1-2V}{1-V}ce}
.
\ees
Thus $P$ and $Q$ have two commons roots if and only if $\breve{P}(\mu_{q'})=0$ 
where $q'\in\{1,2\}\setminus\{q\}$. However
\beas
\breve{P}(\mu_{q'})&=\frac{1-2V}{V^2}\left(bV+(1-V)d-\frac{V(1-V)}{1-2V}e\right)
,\\
&=-\frac{1-2V}{2V^2}((-c+d+e)V-(1-V)b+(-1)^q\sqrt{\Delta_Q}).
\eeas
Since $\Delta_Q=((-c+d+e)V-(1-V)b)^2+4cdV^2$, we have 
$\sqrt{\Delta_Q}>|(-c+d+e)V-(1-V)b|$, hence $\breve{P}(\mu_{q'})\neq 0$. We thus 
have the following lemma
\begin{lemma}\label{notcoprime}
If $V<1/2$, $P$ and $Q$ have at most one common root. This happens if and only 
if $-\frac{1-V}{1-2V}a$ is a root of $Q$, that is
\bes
a=\frac{1-2V}{2V(1-V)}\left((c+d+e)V+(1-V)b+(-1)^q \sqrt{\Delta_Q}\right), 
\qquad q\in\{1,2\},
\ees
where
\beas
\Delta_Q&=(V(c+d+e)+(1-V)b)^2-4V(Vbc+(1-V)(d+e)b),\\
&=((-c+d+e)V-(1-V)b)^2+4cdV^2.
\eeas
In this case, we have $F=\breve{Q}/\breve{P}$, where $\breve{P}$ and 
$\breve{Q}$ are coprime and
\beas
\breve{P}(X)=X^2+\left(-\frac{V}{1-2V}a+b-\frac{1-2V}{1-V}(c+d+e)\right)X+\frac{
1-2V}{1-V}ce,\\
\breve{Q}(X)=V\left(X-\frac{1-V}{1-2V}a+\frac{1-V}{V}b+c+d+e\right).
\eeas
\end{lemma}
We are now ready to seek the solutions of the inverse problem and study the 
identifiability of the system. We first treat the case where $P$ and $Q$ are 
coprime and finish the study with the case where they are not.
\subsection{Coprimality of $P$ and $Q$}
We suppose here that $P$ and $Q$ are coprime. Since $F$ is irreducible and 
$\deg P_x = \deg P =3$ and $\deg Q_x = \deg Q =2$, $F_x=F$ only if the rational 
fraction $F_x$ is irreducible too. In addition, since the leading coefficients 
of $P_x$ and $P$ are identical, as well as those of $Q_x$ and $Q$, $F_x=F$ if 
and only if $P_x=P$ and $Q_x=Q$, that is

\bse\label{eq}\bea
\fl V(\xc+\xd+\xe)+(1-V)\xb=V(c+d+e)+(1-V)b,\label{eq1}\\
\fl V\xc\xe+(1-V)(\xd+\xe)\xb=Vce+(1-V)(d+e)b,\label{eq2}\\
\fl\xa+\xb+\xc+\xd+\xe=a+b+c+d+e,\label{eq3}\\
\fl(\xc+\xd+\xe)\xa+(\xd+\xe)\xb+\xc\xe=(c+d+e)a+(d+e)b+ce,\label{eq4}\\
\fl\xa\xc\xe=ace.
\eea\ese
Doing \eref{eq1}$-V$\eref{eq3} to replace \eref{eq1}, followed by  
$(1-2V)$\eref{eq3}$-$\eref{eq1} to replace \eref{eq3}, then by  
$(1-V)$\eref{eq4}$-$\eref{eq2} to replace \eref{eq4}, the system \eref{eq}  is 
equivalent to
\bse\label{eqq}\bea
\fl-V\xa+(1-2V)\xb=-Va+(1-2V)b,\\
\fl V\xc\xe+(1-V)(\xd+\xe)\xb=Vce+(1-V)(d+e)b,\\
\fl(\xc+\xd+\xe)+\frac{1-V}{1-2V}\xa=(c+d+e)+\frac{1-V}{1-2V}a,\label{eqq3}\\
\fl(\xc+\xd+\xe)\frac{1-V}{1-2V}\xa+\xc\xe= 
(c+d+e)\frac{1-V}{1-2V}a+ce,\label{eqq4}\\
\fl\xc\xe\frac{1-V}{1-2V}\xa=ce\frac{1-V}{1-2V}a.\label{eqq5}
\eea\ese
The equations \eref{eqq3}, \eref{eqq4}, \eref{eqq5} are verified if and only if 
$\frac{\xc+\xd+\xe\pm\sqrt{(\xc+\xd+\xe)^2-4\xc\xe}}{2}$ and 
$\frac{1-V}{1-2V}\xa$ are the roots of the polynomial $R$ of degree $3$
\bes
\fl R(X)=\left(X-\frac{1-V}{1-2V}a\right)\breve{R}(X),\qquad\mbox{where 
}\breve{R}(X)=X^2-(c+d+e)X+ce.
\ees
Hence, the system \eref{eqq} is equivalent to the one obtained by including the 
equation $R\left(\frac{1-V}{1-2V}\xa\right)=0$, although the new system is 
redundant considering equations \eref{eqq3}, \eref{eqq4}, \eref{eqq5} and the 
newly included one. The system \eref{eqq} is then equivalent to
\bse\label{eqqq}\bea
-V\xa+(1-2V)\xb=-Va+(1-2V)b,\\
V\xc\xe+(1-V)(\xd+\xe)\xb=Vce+(1-V)(d+e)b,\\
(\xc+\xd+\xe)+\frac{1-V}{1-2V}\xa=(c+d+e)+\frac{1-V}{1-2V}a,\\
(\xc+\xd+\xe)\frac{1-V}{1-2V}\xa+\xc\xe=(c+d+e)\frac{1-V}{1-2V}a+ce,\\
\xc\xe\frac{1-V}{1-2V}\xa=ce\frac{1-V}{1-2V}a,\label{eqqq5}\\
R\left(\frac{1-V}{1-2V}\xa\right)=0.\label{eqqq6}
\eea\ese
The solution $\xa=a$ to the equation \eref{eqqq6} leads to 
$\xb=b,\xc=c,\xd=d,\xe=e$, that is $M_x=M$. In addition, since 
$0<(-c+d+e)^2+4cd=(c+d+e)^2-4ce<(c+d+e)^2$, the roots of $\breve{R}$ are 
distinct, real and positive. However, for $V>1/2$ and $\xa>0$, we have 
$\frac{1-V}{1-2V}\xa<0$, then $\breve{R}\left(\frac{1-V}{1-2V}\xa\right)\neq 0$ 
and the unique solution of \eref{eqqq6} is $\xa=a$, so that the unique solution 
of the inverse problem is $M_x=M$. We recall that we had the same result for 
$V=1/2$. In addition, we also recall that according to lemma \ref{irpqvoh}, $P$ 
and $Q$ are always coprime for $V\geq1/2$, thus, we have the following lemma
\begin{lemma}
If $V\geq 1/2$, the inverse problem has a unique solution.
\end{lemma}
Suppose now that $V<1/2$, then \eref{eqqq6} has at least two distinct solutions 
since the roots of $\breve{R}$ are distinct and at most three distinct 
solutions. Consider for the moment the solutions of \eref{eqqq} in $\Rbb^5$ 
regardless of their sign and denote by 
$\xa^\pm=\frac{1-2V}{2(1-V)}\left(c+d+e\pm\sqrt{(c+d+e)^2-4ce}\right)$ the 
solutions of $\breve{R}\left(\frac{1-V}{1-2V}\xa\right)=0$. These solutions lead 
to two solutions $(\xapm,\xbpm,\xcpm,\xdpm,\xepm)$ to \eref{eqqq}, provided 
$b+\frac{V}{1-2V}(\xapm-a)\neq0$ and 
$(1-2V)bc-(1-V)(\xapm-a)\left(b-\frac{V}{1-2V}a\right)\neq0$, they are given by
\bse\label{solpm}\bea
\xapm=\frac{1-2V}{2(1-V)} \left(c+d+e\pm\sqrt{(c+d+e)^2-4ce}\right),\\
\xbpm=b+\frac{V}{1-2V}(\xapm-a),\\
\xcpm=\frac{(1-2V)bc-(1-V)(\xapm-a)\left(b-\frac{V}{1-2V}a\right)}{(1-2V)\xbpm},
\\
\xdpm=\frac{bQ\left(-\frac{1-V}{1-2V}a\right)\left(\frac{1-2V}{1-V}
c-\xapm\right)}{(1-2V){\xbpm}^2\xcpm},\\
\xepm=a\left(\frac{1-V}{1-2V}\right)^2\frac{\xamp}{\xcpm}.
\eea\ese
Remark that if both solutions $(\xap,\xbp,\xcp,\xdp,\xep)$ and 
$(\xam,\xbm,\xcm,\xdm,\xem)$ exist, we have
\bes
\xcp\xdp\xcm\xdm=-16cde^2(1-V)^2\frac{\Xi^2}{\Upsilon^2},
\ees
where
\beas
\fl\Xi=a(c+d+e)(1-V)(1-2V)^2-b(d+e)(1-V)(1-2V)^2+ab(1-2V)(1-V)^2\\
-ceV(1-2V)^2-a^2V(1-V)^2,
\eeas
and
\beas
\fl\Upsilon=\left((c+d+e)V(1-2V)+2b(1-V)(1-2V)-2aV(1-V)\right)^2\\
-V^2(1-2V)^2\left((c+d+e)^2-4ce\right)^2.
\eeas
Hence $\xcp\xdp\xcm\xdm\leq0$ and we have $\xcp\xdp\xcm\xdm=0$ if and only if 
$\xdp=0$ or $\xdm=0$, otherwise $\xcp\xdp\xcm\xdm<0$. Moreover, since $\xapm>0$, 
according to equation \eref{eqqq5}, $\xcp$ and $\xep$ have the same sign as well 
as $\xcm$ and $\xem$. In addition, $\xcp$, $\xdp$, $\xep$ are not all negative 
and $\xcm$, $\xdm$, $\xem$ neither, since 
$\frac{\xcpm+\xdpm+\xepm\pm\sqrt{(\xcpm+\xdpm+\xepm)^2-4\xcpm\xepm}}{2}$ are 
roots of the polynomial $R$, thus positive. Hence, if $\xdpm\neq0$, one and only 
one of $\xcp$, $\xdp$, $\xcm$, $\xdm$ is negative so that we have one of the two 
exclusive cases: $\xcp>0$ and $\xdp>0$ or $\xcm>0$ and $\xdm>0$. Consequently, 
the inverse problem has at most two solutions. More precisely, considering 
\eref{solpm} and the previous remark, we have the following lemma
\begin{lemma}\label{locunipqcoprime}
If $V<1/2$ and the polynomials $P$ and $Q$ are coprime,
that is
\bes
a\neq\frac{1-2V}{2V(1-V)}\left((c+d+e)V+(1-V)b\pm\sqrt{\Delta_Q}\right),
\ees
then, we have one of the exclusive cases:
%\ListProperties(Hang=true,Margin2=0.8em,Margin3=2.4em,Margin4=4.8em,Margin5=8em,FinalMark=,Style**={$\!\!$)$\,\,$})
\ListProperties(FinalMark=,Style**={$\!\!$)$\,\,$})
\begin{easylist}
§ \label{1} $\frac{V}{1-2V}\xap\leq\frac{V}{1-2V}a-b$. Then the inverse problem 
has a unique solution: $(a,b,c,d,e)$.
§ $\frac{V}{1-2V}\xam\leq\frac{V}{1-2V}a-b<\frac{V}{1-2V}\xap$ and
§§ \label{2.1} $(1-2V)bc-(1-V)(\xap-a)\left(b-\frac{V}{1-2V}a\right)\leq0$. 
Then the inverse problem has a unique solution: $(a,b,c,d,e)$.
§§ $0<(1-2V)bc-(1-V)(\xap-a)\left(b-\frac{V}{1-2V}a\right)$ and
§§§ \label{2.2.1} 
$\left((-c+d+e)V-(1-V)b\right)^2+4cdV^2<\left((c+d+e)V+(1-V)b-\frac{2V(1-V)}{
1-2V}a\right)^2$. Then the inverse problem has a unique solution: $(a,b,c,d,e)$.
§§§ $\left((c+d+e)V+(1-V)b-\frac{2V(1-V)}{1-2V}
a\right)^2<\left((-c+d+e)V-(1-V)b\right)^2+4cdV^2$ and
§§§§ \label{2.2.2.1} $a=\xap$. Then the inverse problem has a unique solution: 
$(a,b,c,d,e)$.
§§§§ $a\neq\xap$. Then the inverse problem has two solutions: $(a,b,c,d,e)$ and 
$(\xap,\xbp,\xcp,\xdp,\xep)$.
§ $\frac{V}{1-2V}a-b<\frac{V}{1-2V}\xam$ and
§§ $b=\frac{aV(1-V)(\xap-a)}{(1-2V)\left((1-V)(\xap-a)-(1-2V)c)\right)}$ and
§§§ \label{3.1.1} $b-\frac{V}{1-2V}a\leq0$. Then the inverse problem has a 
unique solution: $(a,b,c,d,e)$.
§§§ $0<b-\frac{V}{1-2V}a$ and
§§§§ 
$\left((-c+d+e)V-(1-V)b\right)^2+4cdV^2<\left((c+d+e)V+(1-V)b-\frac{2V(1-V)}{
1-2V}a\right)^2$ and
§§§§§ \label{3.1.2.1.1} $a=\xam$. Then the inverse problem has a unique 
solution: $(a,b,c,d,e)$.
§§§§§ $a\neq\xam$. Then the inverse problem has two solutions: $(a,b,c,d,e)$ 
and $(\xam,\xbm,\xcm,\xdm,\xem)$.
§§§§ \label{3.1.2.2} $\left((c+d+e)V+(1-V)b-\frac{2V(1-V)}{1-2V}
a\right)^2<\left((-c+d+e)V-(1-V)b\right)^2+4cdV^2$. Then the inverse problem 
has a unique solution: $(a,b,c,d,e)$.
§§ $b=\frac{aV(1-V)(\xam-a)}{(1-2V)\left((1-V)(\xam-a)-(1-2V)c)\right)}$ and
§§§ \label{3.2.1} $0\leq b-\frac{V}{1-2V}a$. Then the inverse problem has a 
unique solution: $(a,b,c,d,e)$.
§§§ $b-\frac{V}{1-2V}a<0$ and
§§§§ \label{3.2.2.1} 
$\left((-c+d+e)V-(1-V)b\right)^2+4cdV^2<\left((c+d+e)V+(1-V)b-\frac{2V(1-V}{1-2V
}a\right)^2$. Then the inverse problem has a unique solution: $(a,b,c,d,e)$.
§§§§ $\left((c+d+e)V+(1-V)b-\frac{2V(1-V}{1-2V}
a\right)^2<\left((-c+d+e)V-(1-V)b\right)^2+4cdV^2$ and
§§§§§ \label{3.2.2.2.1} $a=\xap$. Then the inverse problem has a unique 
solution: $(a,b,c,d,e)$.
§§§§§ $a\neq\xap$. Then the inverse problem has two solutions: $(a,b,c,d,e)$ 
and $(\xap,\xbp,\xcp,\xdp,\xep)$.
§§ $b\neq\frac{aV(1-V)(\xapm-a)}{(1-2V)\left((1-V)(\xapm-a)-(1-2V)c)\right)}$ 
and
§§§ 
$\left((c+d+e)V+(1-V)b-\frac{2V(1-V}{1-2V}
a\right)^2<\left((-c+d+e)V-(1-V)b\right)^2+4cdV^2$ and 
$0<(1-2V)bd-(1-V)(\xap-a)\left(b-\frac{V}{1-2V}a\right)$ and
§§§§ \label{3.3.1.1} $a=\xap$. Then the inverse problem has a unique solution: 
$(a,b,c,d,e)$.
§§§§ $a\neq\xap$. Then the inverse problem has two solutions: $(a,b,c,d,e)$ and 
$(\xap,\xbp,\xcp,\xdp,\xep)$.
§§§ 
$\left((-c+d+e)V-(1-V)b\right)^2+4cdV^2<\left((c+d+e)V+(1-V)b-\frac{2V(1-V}{1-2V
}a\right)^2$ and $0<(1-2V)bd-(1-V)(\xam-a)\left(b-\frac{V}{1-2V}a\right)$ and
§§§§ \label{3.3.2.1} $a=\xam$. Then the inverse problem has a unique solution: 
$(a,b,c,d,e)$.
§§§§ $a\neq\xam$. Then the inverse problem has two solutions: $(a,b,c,d,e)$ and 
$(\xam,\xbm,\xcm,\xdm,\xem)$.
\end{easylist}
\end{lemma}
\subsection{Non-coprimality of $P$ and $Q$}
We suppose here that $P$ and $Q$ are not coprime. Then, according to lemma 
\ref{irpqvoh}, $V<1/2$, moreover, according to lemma \ref{notcoprime}, 
$-\frac{1-V}{1-2V}a$ is a root of $Q$, that is
\bes
a=\frac{1-2V}{2V(1-V)}\left((c+d+e)V+(1-V)b+(-1)^q \sqrt{\Delta_Q}\right), 
\qquad q\in\{1,2\},
\ees
where
\beas
\Delta_Q&=(V(c+d+e)+(1-V)b)^2-4V((1-V)(d+e)b+Vce),\\
&=((-c+d+e)V-(1-V)b)^2+4cdV^2.
\eeas
Moreover, $F=\breve{Q}/\breve{P}$ where $\breve{P}$ and $\breve{Q}$ are coprime 
and
\bes
\breve{P}(X)=X^2+\left(-\frac{1-2V}{1-V}(c+d+e)+b-\frac{V}{1-2V}a\right)X+\frac{
1-2V}{1-V}ce,
\ees
\bes
\breve{Q}(X)=V\left(X+c+d+e+\frac{1-V}{V}b-\frac{1-V}{1-2V}a\right).
\ees
$F_x=F$ only if $P_x$ and $Q_x$ are not coprime, that is $-\frac{1-V}{1-2V}\xa$ 
is a root of $Q_x$
\bes
\fl\xa=\frac{1-V}{2V(1-2V)}\left((\xc+\xd+\xe)V+(1-V)\xb+(-1)^{q_x} 
\sqrt{\Delta_{Q_x}}\right), \qquad q_x\in\{1,2\},
\ees
where
\beas
\fl\Delta_{Q_x}&=(V(\xc+\xd+\xe)+(1-V)\xb)^2-4V((1-V)(\xd+\xe)\xb+V\xc\xe),\\
\fl&=((-\xc+\xd+\xe)V-(1-V)\xb)^2+4\xc\xd V^2.
\eeas
So that, $F_x=\breve{Q}_x/\breve{P}_x$ where $\breve{P}_x$ and $\breve{Q}_x$ 
are coprime and
\beas
\fl\breve{P}_x(X)=X^2+\left(-\frac{1-2V}{1-V}(\xc+\xd+\xe)+\xb-\frac{V}{1-2V}
\xa\right)X+\frac{1-2V}{1-V}\xc\xe,\\
\fl\breve{Q}_x(X)=V\left(X+\xc+\xd+\xe+\frac{1-V}{V}\xb-\frac{1-V}{1-2V}\xa\right).
\eeas
Since $\breve{Q}_x/\breve{P}_x$ and $\breve{Q}/\breve{P}$ are irreducible and 
the leading coefficients of $\breve{P}_x$ and $\breve{P}$ are identical, as well 
as those of $\breve{Q}_x$ and $\breve{Q}$, $F_x=F$ only if 
$\breve{P}_x=\breve{P}$ and $\breve{Q}_x=\breve{Q}$. Hence, $F_x=F$ if and only 
if
\bse
\bea
\fl V\left(\frac{1-V}{1-2V}\xa\right)^2+((1-V)(\xd+\xe)\xb+V\xc\xe)\nonumber\\
-(V(\xc+\xd+\xe)+(1-V)\xb)\frac{1-V}{1-2V}\xa=0,\\
\xc\xe=ce,\\
\fl-\frac{1-2V}{1-V}(\xc+\xd+\xe)+\xb\nonumber\\
-\frac{V}{1-2V}\xa=-\frac{1-2V}{1-V}(c+d+e)+\xb-\frac{V}{1-2V}a,\label{eqnc3}\\
\fl V(\xc+\xd+\xe)+(1-V)\xb\nonumber\\
-\frac{V(1-V)}{1-2V}\xa=V(c+d+e)+(1-V)b-\frac{V(1-V)}{1-2V}a.\label{eqnc4}
\eea
\ese
Hence, doing \eref{eqnc4}$-(1-V)$\eref{eqnc3} to replace \eref{eqnc3}, followed 
by $(1-V)$\eref{eqnc4}-$V$\eref{eqnc3} to replace \eref{eqnc4}, we get
\bse
\bea
\fl V\left(\frac{1-V}{1-2V}\xa\right)^2+((1-V)(\xc+\xe)\xb+V\xd\xe)\nonumber\\
-(V(\xc+\xd+\xe)+(1-V)\xb)\frac{1-V}{1-2V}\xa=0,\label{eqqnc1}\\
\xc\xe=ce,\\
\xc+\xd+\xe=c+d+e,\\
(1-2V)\xb-V\xa=(1-2V)b-Va.\label{eqqnc4}
\eea
\ese
Replacing $\xa$ in \eref{eqqnc1} by $a+\frac{1-2V}{V}(\xb-b)$ given by 
\eref{eqqnc4}, we get
\bse
\bea
(\xc+h)\xb=\frac{V}{1-V}(h^2+(\xc+\xd+\xe)h+\xc\xe),\label{eqqqnc1}\\
\xc\xe=ce,\\
\xc+\xd+\xe=c+d+e,\\
(1-2V)\xb-V\xa=(1-2V)b-Va,
\eea
\ese
where $h=\frac{1-V}{V}b-\frac{1-V}{1-2V}a$. Then, replacing $\xc+\xd+\xe$ 
by $c+d+e$ and $\xc\xe$ by $ce$ in \eref{eqqqnc1}, we get
\bse\label{eqqqqnc}
\bea
\left(\xc+\frac{1-V}{V}b-\frac{1-V}{1-2V}a\right)\xb= 
\left(c+\frac{1-V}{V}b-\frac{1-V}{1-2V}a\right)b,\label{eqqqqnc1}\\
\xc\xe=ce,\label{eqqqqnc2}\\
\xc+\xd+\xe=c+d+e,\label{eqqqqnc3}\\
(1-2V)\xb-V\xa=(1-2V)b-Va.\label{eqqqqnc4}
\eea
\ese
The set of solutions $(\xa,\xb,\xc,\xd,\xe)$ to \eref{eqqqqnc} is a curve given 
by the intersection of two cylindrical hypersurfaces over hyperbolas 
(\ref{eqqqqnc1},~\ref{eqqqqnc2}) and two hyperplanes 
(\ref{eqqqqnc3},~\ref{eqqqqnc4}). If $c+\frac{1-V}{V}b-\frac{1-V}{1-2V}a=0$, the 
cylindrical hypersurface given by \eref{eqqqqnc1} degenerates into two 
hyperplanes whose respective equations are 
$\xb=0$ and $\xc=c$ so that the set of solutions of \eref{eqqqqnc} is composed 
of the line of equations $(1-2V)\xb-V\xa=(1-2V)b-Va,\xc=c,\xd=d,\xe=e$ and the 
hyperbola of equations $\xa=0,\xb=0,\xc+\xd+\xe=c+d+e,\xc\xe=ce$. Considering 
the inverse problem, we thus have the following lemma 
\begin{lemma}
If $V<1/2$ and the polynomials $P$ and $Q$ are not coprime
\beas
\fl P(X)=X^3+(a+b+c+d+e)X^2+(ac+ad+ae+bd+be+ce)X+ace,\\
\fl Q(X)=VX^2+(V(c+d+e)+(1-V)b)X+(Vce+(1-V)(d+e)b),
\eeas
then, if $c+\frac{1-V}{V}b-\frac{1-V}{1-2V}a=0$, the set of solutions 
$(\xa,\xb,\xc,\xd,\xe)$ to the inverse problem is given by the open half-line of 
$\Rbb^5$ of parametric equation
\bes
t\mapsto\left(a+(1-2V)t,b+Vt,c,d,e\right),\qquad 
t\in\left(-\frac{b}{V},+\infty\right).
\ees
Otherwise, the set of solutions $(\xa,\xb,\xc,\xd,\xe)$ to the inverse problem 
is given by the curve of $\Rbb^5$ of parametric equation
\bes
\fl t \mapsto 
\left(a+\frac{(1-2V)(c-t)b}{Vt+(1-V)b-\frac{V(1-V)}{1-2V}a},\frac{c+\frac{1-V}{V
}b-\frac{1-V}{1-2V}a}{t+\frac{1-V}{V}b-\frac{1-V}{1-2V}a}b,t,c+d+e-t-\frac{ce}{t
},\frac{ce}{t}\right),
\ees
where
\bes
\fl t\in
\cases{\left(\beta^-,\min\left(\beta^+,\frac{1-V}{1-2V}a+\frac{(1-2V)bc}{
(1-2V)b-Va}\right)\right),&\mbox{if }$Va-(1-2V)b<0$,\\
\left(\max\left(\frac{1-V}{1-2V}a-\frac{1-V}{V}b,\beta^-\right),\beta^+\right),
&\mbox{if }$0\leq Va-(1-2V)b<(1-2V)Vc$,\\
\left(\beta^-,\min\left(\beta^+,\frac{1-V}{1-2V}a-\frac{1-V}{V}b\right)\right),
&\mbox{if }$(1-2V)Vc<Va-(1-2V)b$,}
\ees
where
\bes
\beta^\pm=\frac{c+d+e\pm\sqrt{(c+d+e)^2-4ce}}{2}.
\ees
\end{lemma}
\end{document}